\documentclass[12pt]{article}
\usepackage[utf8]{inputenc}
\usepackage{amsmath, amsthm, amssymb, amsfonts}
\usepackage{import}
\usepackage[numbers, comma]{natbib}
\usepackage{fullpage}
\usepackage{authblk}
\usepackage{apxproof}
\usepackage{import}
\usepackage{mathtools}
\usepackage{tikz}
\usepackage{float}
\usepackage{enumitem}

\usetikzlibrary{calc,patterns,angles,shapes, positioning, intersections, quotes}

\newtheorem{theorem}{Theorem}
\newtheoremrep{thm}{Theorem}
\newtheoremrep{lem}{Lemma}

\usepackage{graphicx}
\newtheorem{corollary}[theorem]{Corollary}

\newtheorem{conjecture}[theorem]{Conjecture}

\theoremstyle{definition}
\newtheorem{definition}{Definition}[section]

\usepackage[colorlinks=false,pagebackref=true, hidelinks]{hyperref}

\allowdisplaybreaks

\begin{document}
\title{Project selection with partially verifiable information\thanks{We are grateful to Laura Doval, Federico Echenique and Amit Goyal for helpful comments and suggestions. We would also like to thank seminar audiences at Caltech and Delhi School of Economics for valuable feedback. }}

\author{Sumit Goel\thanks{California Institue of Technology; sgoel@caltech.edu; 0000-0003-3266-9035}  \quad Wade Hann-Caruthers\thanks{California Institue of Technology; whanncar@gmail.com; 0000-0002-4273-6249}}

\date{}

\maketitle              





\begin{abstract}
    We consider a principal agent project selection problem with asymmetric information. There are $N$ projects and the principal must select exactly one of them. Each project provides some profit to the principal and some payoff to the agent and these profits and payoffs are the agent's private information. We consider the principal's problem of finding an optimal mechanism for two different objectives: maximizing expected profit and maximizing the probability of choosing the most profitable project. Importantly, we assume partial verifiability so that the agent cannot report a project to be more profitable to the principal than it actually is. Under this no-overselling constraint, we characterize the set of implementable mechanisms. Using this characterization, we find that in the case of two projects, the optimal mechanism under both objectives takes the form of a simple cutoff mechanism. The simple structure of the optimal mechanism also allows us to find evidence in support of the well-known ally-principle which says that principal delegates more authority to an agent who shares their preferences.
\end{abstract}
\section{Introduction}

Suppose a principal has to chose exactly one of $N$ available projects but does not know how profitable they are. There is an agent who is fully informed about these profits but also has its own preference over the projects. The principal would like to use the agent's information and chose a profitable project. Assuming the principal has commitment power and cannot use transfers to incentivize the agent, we consider the problem of finding the optimal mechanism for the principal in this project selection framework.  \\  

Let's first quickly consider the standard setting where the agent can lie arbitrarily. So for any principal and agent payoff vectors $(p,a) \in \Theta$, the agent can report any $(\pi, \alpha) \in \Theta$.  Consider any mechanism $d$ and suppose its range is $S \subset [N]$. Under this mechanism, the agent will always report a $(\pi, \alpha)$ so that its favorite project in $S$ defined by $\arg max_{i \in S} a_i$ is chosen. Thus, the principal can only fix a set $S \subset [N]$ and commit to choosing agent's favorite project in $S$. This means that if the payoffs $p_i, a_i$ are i.i.d. across projects, all mechanisms lead to expected payoff $\mathbb{E}[p_i]$ for the principle. And the probability of the principal choosing the best project is $\frac{1}{N}$ for any mechanism. Thus, the principle may as well choose a constant mechanism and commitment power doesn't buy anything for the principal. \\

An important assumption in the above setting is that the agent can lie arbitrarily. But the agent's ability to manipulate may be limited if it is required to support its claims with some form of evidence. For instance, if a tech firm wants to a hire a programmer and the hiring committee is biased towards candidates with better social skills, the firm can require the hiring committee to provide certificates that support the reported coding skills of the candidates. Now the hiring committee can still potentially hide certificates and understate the coding ability of an applicant, but it cannot furnish fake certificates and overstate its coding abilities. Thus, by requiring some kind of supporting evidence, the principal can constraint the kind of manipulations that the agent can make. \\

In this paper, we consider a setting where the agent cannot oversell a project to the principal. That is, the agent cannot say a project is more profitable for the principal than it actually is. So if the true state is $(p,a) \in \Theta$, the agent's message space is  $M(p,a)=\{(\pi, \alpha) \in \Theta: \pi_i \leq p_i\ \text{ for all } i \} $. Since the set of manipulations that the principal has to guard against is now smaller, the class of truthful mechanisms is potentially bigger. Note that our message correspondence  satisfies the nested range condition ($\theta' \in M(\theta) \implies M(\theta') \subset M(\theta)$) from \citet{green1986partially} and thus, we can without loss of generality restrict attention to truthful mechanisms.\\

Under the partial verifiability constraint of no overselling-constraint, we first characterize the class of truthful mechanisms and call them table mechanisms. A table mechanism is defined by an increasing set function which determines the set of projects on the table as a function of the reported profit values $\pi$. The mechanism then chooses the agent's favorite project according to the reported $\alpha$ from those on the table. Thus, the class of truthful mechanisms is now significantly bigger. We use this characterization to find the optimal mechanism for the principal under two different objectives.\\

First, we consider the objective of maximizing the expected profit for the principal. For the case of two projects, we show that the optimal table mechanism is a simple cutoff mechanism. In this mechanism, one project is always on the table and the other project is on the table if its reported profit meets a cutoff that depends on the bivariate distribution $F$ from which the principal and agent payoffs $(p_i,a_i)$ are drawn for each project. When the payoffs are independent, this cutoff equals the expected profit. We also discuss the well-known ally principle in our framework by considering the case where $F$ is bivariate normal. In this case, we find that the optimal cutoff is decreasing in the correlation between the principal and agent payoffs and thus, the principal lends more leeway to an agent who shares their preferences.  For the case of $N>2$ projects, we obtain a relatively weak result. Assuming $F$  is uniform on $[0,1]^2$, we show that a single cutoff mechanism is optimal in the simple subclass of cutoff mechanisms. Next, we also consider the objective of maximizing the probability of choosing the best project. Again, for the case of two projects, we show that a cutoff mechanism is optimal and when the project payoffs are independent, the optimal cutoff equals the median of the principal's profit. \\

\subsection{Related literature}

Mechanism design has often been used to deal with problems of asymmetric information ( \citet{myerson1981optimal}, \citet{myerson1983efficient}).  An important theme in the literature on mechanism design is characterizing the set of implementable mechanisms (\citet{gibbard1973manipulation}, \citet{satterthwaite1975strategy}, \citet{dasgupta1979implementation}, \citet{green1986partially}). In particular,  \citet{green1986partially} introduce  the idea of partially verifiable information in mechanism design and identify a necessary and sufficient condition, called the ``Nested Range Condition",  under which the set of implementable mechanisms coincides with the set of truthfully implementable mechanisms (i.e. the revelation principle holds). Later research focuses on identifying implementability conditions in other  environments with partially verifiable information; \citet{kartik2009strategic} looks at Nash implementation, \citet{deneckere2008mechanism} considers lying with finite costs, \citet{ben2012implementation} and \citet{singh2001implementation} allow for transfers,  and  \citet{caragiannis2012mechanism} considers probabilistic verification. Some computer scientists have looked at the trade-off between monetary transfers and partial verifiability in terms of implementing social choice functions (\citet{ferraioli2016verify}, \citet{fotakis2017combinatorial}). Our setup belongs to the environment considered by Green and Laffont and satisfies their ``Nested Range Condition". Thus, without loss of generality, we restrict attention to truthfully implementable mechanisms in our analysis.\\

Our paper contributes to the literature finding optimal or efficient mechanisms in  environments with a specific form of partial verifiability (\citet{maggi1995costly}, \citet{lacker1989optimal},  \citet{moore1984global}, \citet{celik2006mechanism}). For instance,     \citet{munro2014hide} argues using a model in which only some expenditure can be hidden that spouses hiding income and assets from one another is efficient. \citet{deneckere2001mechanism} explains the complicated selling practices of real-world monopolists by considering an economy where some agents have limited ability to misrepresent their preferences. This paper considers the partial verifiability constraint of no-overselling and potentially explains the use of cutoff mechanisms in settings that only admit positive evidence which can be hidden but not fabricated.  \\

Our work also relates to the literature studying principal agent project selection problems with different modeling assumptions (\citet{ben2014optimal}, \citet{mylovanov2017optimal},  \citet{armstrong2010model} and \citet{guo_project_nodate}).  \citet{ben2014optimal} and \citet{mylovanov2017optimal} consider a problem where the principal has to choose one of $N$ agents  who prefer being chosen and provide some private value to the principal from being chosen. In \citet{ben2014optimal}, the principal can verify his value from agent $i$ at a cost $c_i$, while \citet{mylovanov2017optimal} assumes ex-post verifiability so that the principal can penalize the winner by destroying a certain fraction of the surplus.   \citet{armstrong2010model} considers a project delegation problem in which the principal can verify characteristics of the chosen project but is uncertain about the set of available projects. These papers find their respective optimal mechanisms for the principal and call them the  favored agent mechanism \citet{ben2014optimal}, shortlisting procedure  \citet{mylovanov2017optimal}, and the threshold rule  \citet{armstrong2010model}. While these mechanisms have some flavor of the cutoff mechanisms we obtain in this paper, there are important differences in the setup we consider here.  Primarily, in their setups, the principal is empowered by ex-post verifiability of the reported values and the ability to use a prohibitively high punishment to deter the agent from telling \textit{any} lie, whereas in our setup, the agent is constrained in that he cannot oversell, but the principal does not have the power to directly deter the agent from underselling. 
 \\
 
The paper proceeds as follows. In section 2, we present the model and definitions. Section 3 characterizes the class of truthful mechanisms. In section 4 and 5, we consider the two different objectives of maximizing expected profit and maximizing probability of choosing the best project for the principal. In section 6, we give some remarks and Section 7 concludes. The more technical proofs are in the appendix.\\

\section{Model}
There are two parties: a principal and an agent. The principal has a set of available projects $[N]=\{1,2,\dots, N\}$ and must choose one of them. Each project $i$ leads to payoffs $(p_i,a_i) \in X \subset \mathbb{R}^2$ where $p_i$ denotes the profit for the principal and $a_i$ is the utility to the agent. The payoffs $(p_i,a_i)$ are i.i.d. from a bi-variate distribution $F$ and this is all the information that the principal has. The agent knows the true payoffs from all the projects $(p, a) \in X^n=\Theta$.\\
 
We assume that the principal can commit to a mechanism $d:\Theta \to [N]$ so that if the agent reports payoffs $(\pi, \alpha)$ when the true state is $(p,a)$, the project $d(\pi, \alpha)$ is chosen leading to final payoffs $p_{d(\pi, \alpha)}, a_{d(\pi, \alpha)}$ for the principal and agent respectively. \\

As discussed earlier, if the agent can lie arbitrarily, the principal cannot do better than by choosing a project at random. So we assume a natural partial verifiability constraint of no overselling under which the agent cannot report a project to be more profitable than it actually is. Such a constraint on the message space may be inherent in the environment or induced by the principal by requiring the agent to furnish some kind of evidence supporting its claims. Formally, the agent's state dependent message space takes the form:
\begin{align*}
    M(p,a)=\{(\pi, \alpha) \in \Theta : \pi_i \leq p_i \text{ } \forall i \in [N] \}
\end{align*}

Since our message space satisfies the Nested Range condition of \citet{green1986partially}, $$\theta' \in M(\theta) \implies M(\theta') \subset M(\theta)$$ we can without loss of generality restrict attention to truthful mechanisms. 
 
\begin{definition}
A mechanism $d$ is \textit{truthful} if for any $(p,a)$ and $(\pi,\alpha) \in M(p,a) $
\begin{align*}
    a_{d(p, a)} \geq a_{d(\pi, \alpha)}.
\end{align*}
\end{definition}

We'll consider the principal's problem of finding the optimal truthful mechanism for two different objectives:
\begin{itemize}
    \item Maximizing the expected profit: $$\max_{d: d \text{ is truthful}} \mathbb{E}[p_{d(p,a)}]$$
    \item Maximizing the probability of choosing the best project: $$\max_{d: d \text{ is truthful }} \mathbb{P}[d(p,a) \in \arg max_i p_i]$$
\end{itemize}

First, we characterize the class of truthful mechanisms.

\section{Characterization of truthful mechanisms}

We begin by defining a special class of mechanisms which we call table mechanisms.

\begin{definition}
A mechanism $d$ is a table mechanism if there exists a function $f: \Theta \rightarrow 2^{N}$ with the properties
\begin{itemize}
\item  $f(p, a) \neq \phi$
\item $f(p, a)=f\left(p, a^{\prime}\right)=f(p)$
\item  $p \leq p^{\prime} \Longrightarrow f(p) \subset f\left(p^{\prime}\right)$
\end{itemize}
so that
 $$d(p, a) \in \arg \max _{i}\left\{a_{i}: i \in f(p)\right\} .$$

\end{definition}

In words, a table mechanism is defined by an increasing set function that determines the set of projects on the table and the mechanism always chooses the agent's favorite project among those on the table.  The first condition says that there is a default project which is always on the table. The second condition says that the set of projects on the table cannot depend on the agent's payoffs. The third condition is that the set of projects on the table weakly increases as the profit vector increases. 

\begin{theorem}
\label{thm:table}
$d$ is truthful if and only if it is a table mechanism.
\end{theorem}

It is fairly straightforward from the definitions to check that table mechanisms are truthful. Indeed, under a table mechanism, the agent prefers reporting higher $\pi$'s to reporting lower ones and, given any report of $\pi$'s, prefers reporting her payoffs to misreporting her payoffs (since such a misrepresentation can only lead the principal to make a choice which gives the agent a lower payoff.) The other direction is more involved.

\begin{proof}[Proof of Theorem \ref{thm:table}]
Suppose $d$ is a table mechanism. Consider any profile $(p,a)$. If the agent reports some $(\pi,\alpha)$, we know that $\pi \leq p$ from the constraint $(\pi, \alpha) \in M(p,a)$. Since $f$ is increasing, it follows that $f(\pi)  \subset f(p)$. Since the agent gets his preferred project among those available and reporting truthfully maximizes his set of available projects, the agent cannot gain by misreporting. Therefore, $d$ is truthful.\\

Now suppose that $d$ is a truthful mechanism. Define the function $f: \Theta \to 2^N$ so that $i \in f(p,a)$ if and only if there exists some $(p',a') \in M(p,a)$ such that $d(p',a')=i$. First, we will show that the function satisfies the three properties in the definition of table mechanism:

\begin{itemize}
\item Observe that $(p,a) \in M(p,a)$ and thus, $d(p,a) \in f(p,a) \implies  f(p,a) \neq \phi$ for any $(p,a) \in \Theta$

\item The property that $f$ does not depend on agent payoffs follows from observing that $M(p,a)=M(p,a')$. Thus, if $i \in f(p,a)$, $i \in f(p,a')$ and vice versa. Thus, we have $f(p,a)=f(p,a')$ for any $(p,a), (p,a')$. 

\item Take any $p,p'$ such that $p \leq p'$. Suppose $i \in f(p)$ which implies that there exists $(\pi, \alpha) \in M(p,a)$ with $\pi \leq p$ so that $d(\pi, \alpha)=i$. But then, $(\pi, \alpha) \in M(p',a)$ as well and so $i \in f(p')$. Thus, we get that $f(p) \subset f(p')$. 
\end{itemize}

Now we want to show that for any state $(p,a)$,  $d(p,a) \in \arg \max _{i}\left\{a_{i}: i \in f(p)\right\}$

Suppose towards a contradiction that $d(p, a)$ is not in this set. By definition, $d(p,a) \in f(p)$. Let $j \in \arg \max _{i}\left\{a_{i}: i \in f(p)\right\}$. Then $a_j > a_{d(p,a)}$ and $j \in f(p)$. But the fact that $j \in f(p)$  implies that there exists a $(\pi, \alpha) \in M(p,a)$ such that $d(\pi,\alpha)=j$. But then, the agent can misreport at state $(p,a)$ to $(\pi, \alpha)$ and gain from this manipulation. This contradicts the fact that $d$ is truthful and so it must be that $d(p,a) \in \arg \max _{i}\left\{a_{i}: i \in f(p)\right\}$. It follows then that $d$ is a table mechanism. 
\end{proof}

For simplicity going forward, we will assume (without loss of generality) that $N \in f(p)$ for all $(p,a) \in \Theta$. That is, in a table mechanism, project $N$ is always on the table.\\

Before discussing the results, we define a subclass of table mechanisms that take a simple cutoff form.

\begin{definition}
A mechanism $d$ is a cutoff mechanism if it is a table mechanism and for $i=1, \dots, N-1$, there exist cutoffs $c_i \in X$, such that $i \in f(p)$ if and only if $p_i \geq c_i$.
\end{definition}

In a cutoff mechanism, a project is on the table if the principal's profit from the project meets a threshold. That is, whether a project is on the table or not depends only on that particular project's profit value. The principal then chooses the agent's favorite project among those that meet the cutoff and the default project.

\section{Maximizing expected profit}

In this section, we consider the principal's problem of finding the optimal table mechanism $d$ for maximizing expected profit $\mathbb{E}[p_{d(p,a)}]$. For the most part, we'll consider and solve the problem for the case of $N=2$ projects. We'll briefly discuss the case of $N>2$ projects towards the end of this section.

\subsection{2 projects}
In the case of two projects, we have $(p_1,a_1) \sim F$ and $(p_2,a_2)\sim F$. In a table mechanism $f$ with 2 projects, we can assume without loss of generality that $2 \in f(p)$ for all $p$ and so the principal only really has to decide the set of vectors when $1 \in f(p)$.

\begin{theorem}
\label{thm:n=2_optimal_is_cutoff}

For two projects with $(p_i,a_i) \sim F$, the optimal truthful mechanism is a cutoff mechanism. The optimal cutoff $c_1$ is defined by
$$\mathbb{E}\left[(p_1-p_2)\Pr(a_1>a_2\vert p_2)\vert p_1=c_1\right]=0$$
\end{theorem}

\begin{proof}
Suppose $d$ is truthful with associated function $f$. From Theorem \ref{thm:table}, we know that $d$ is a table mechanism. So $2 \in f(p)$ for all $p_1,p_2$. Define $c=\sup \{p_1: 1 \notin f(p_1,p_1)\}$. Define the cutoff mechanism $d'$ so that $1 \in f'(p) \iff p_1 \geq c_1$. We'll show that the expected profit for the principal from $d'$ is at least as high as the expected profit from $d$. In fact, we will show that this holds conditional on $(p_1,p_2)$, and hence in expectation.




Consider the following (exhaustive and mutually exclusive) cases depending on whether $p_i \geq c$ or $p_i<c$:

\begin{itemize}
    \item $p_1 \in (-\infty,c), p_2 \in (-\infty,c)$: For any such $p_1,p_2$, we know both $f(p)=f'(p)=\{2\}$ and therefore, the second project is chosen for all such profiles. Thus, the two mechanisms are identical and generate the same profit for the principal in this case. 
    
    \item $p_1 \in (\infty,c), p_2 \in [c,\infty)$: In this case, $f'(p)=\{2\}$ and thus project 2 is chosen for sure. Note that the principal prefers project 2 over 1 in these profiles and thus, the profit from the cutoff mechanism is weakly higher for any such $p_1,p_2$. 
    
    \item $p_1 \in [c,\infty), p_2 \in (-\infty,c)$:
    Now $f'(p)=\{1,2\}$ while $f(p)$ can be either $\{2\}$ or $\{1,2\}$. Observe that the principal strictly prefers project 1 over 2 in all these profiles. Thus, the cutoff mechanism again leads to weakly higher profits for such $p_1,p_2$. 
    
    \item $p_1 \in [c,\infty), p_2 \in [c,\infty)$: Here we have $f(p)=f'(p)=\{1,2\}$. Thus, the two mechanisms are identical in this set and lead to same profits for the principal.
    
\end{itemize}

This shows that for any truthful mechanism, there is a cutoff mechanism under which the principal's expected profit is weakly higher. Thus, the optimal truthful mechanism must be a cutoff mechanism. Now our problem is just to find the optimal cutoff $c$.

Consider the decision problem of the principal for any given $p_1,p_2,a_1,a_2$. It can either 
\begin{itemize}
    \item not make project 1 available and  get $p_2$
    \item make project 1 available and get $p_1 \mathbb{I}_{a_1 \geq a_2}+p_2\mathbb{I}_{a_2 >a_1}$ 
\end{itemize}

The constraint imposed by truthfulness and optimality of cutoff mechanisms imply that the principal can only base this decision on the value of $p_1$. Thus, taking expectation with respect to $p_2,a_1,a_2$, we get that the two alternatives are:
\begin{itemize}
    \item not make project 1 available and  get $\mathbb{E}[p_2|p_1]$
    \item make project 1 available and get $\mathbb{E}\left[p_1 \mathbb{I}_{a_1 \geq a_2}+p_2\mathbb{I}_{a_2 >a_1}|p_1\right]$ 
\end{itemize}

Thus, the principal would want to make project 1 available if and only if
\begin{align*}
   &\mathbb{E}\left[p_1 \mathbb{I}_{a_1 \geq a_2}+p_2\mathbb{I}_{a_2 >a_1}|p_1\right] \geq \mathbb{E}[p_2|p_1]\\
   \iff & \mathbb{E}\left[p_1 \mathbb{I}_{a_1 \geq a_2}|p_1\right] \geq \mathbb{E}\left[p_2 \mathbb{I}_{a_1 \geq a_2}|p_1\right] \\
   \iff & \mathbb{E}\left[(p_1-p_2) \mathbb{I}_{a_1 \geq a_2}|p_1\right] \geq 0 \\
   \iff & \mathbb{E}\left[(p_1-p_2) \mathbb{P}[a_1 \geq a_2|p_2]|p_1\right] \geq 0 \\
\end{align*}

At the optimal cutoff, the principal should be indifferent between the two alternatives and so the cutoff $c_1$ is defined by the solution  to the equation $$\mathbb{E}\left[(p_1-p_2) \mathbb{P}[a_1 \geq a_2|p_2]|p_1=c_1\right] = 0 $$
\end{proof}

In the special case where the principal and agent payoffs are independent, we get the following corollary.
\begin{corollary}
Suppose $F$ is such that the principal and agent payoffs $(p_i,a_i)$ are independent. Then the optimal cutoff is given by $c_1=\mathbb{E}[p_1]$.
\end{corollary}

\subsection{Ally principle}
In this subsection, we discuss the implications of our model for the well-known Ally principle which states that a principal delegates more authority to an agent with more aligned preferences. For this purpose, we assume that the principal agent payoffs for each project is bivariate normal $(p_i,a_i) \sim N(0,0,1,1,\rho)$  and are drawn i.i.d. across projects. Now the question is whether we can say something systematic about $c(\rho)$, the optimal cutoff as a function of the correlation $\rho$.

\begin{theorem}
\label{cutoff1}
For $N=2$ projects with $(p_i, a_i) \sim N(0,0,1,1,\rho)$, the optimal cutoff is defined by the equation $$c\Phi(tc)+t\phi(tc)=0$$ where $t=\frac{\rho}{\sqrt{2-\rho^2}}$. The optimal cutoff is decreasing in $\rho$. 
\end{theorem}



\begin{appendixproof}
\textit{of Theorem ~\ref{cutoff1}}

We know that the optimal cutoff $c(\rho)$ is the solution to the following equation 
$$\mathbb{E}\left[(p_1-p_2) \mathbb{P}[a_1 \geq a_2|p_2]|p_1=c(\rho)\right] = 0 $$

Let us simplify the above expression. First, we want to find $\mathbb{P}[a_1 \geq a_2|p_1, p_2]$. We know that if $X,Y \sim N(\mu_x,\mu_y,\sigma^2_x, \sigma^2_y,\rho)$, then the conditional distribution $$X \mid Y \sim N\left(\mu_{x}+\rho \frac{\sigma_{x}}{\sigma_{y}}\left(y-\mu_{y}\right), \sigma^2_{x} (1-\rho^{2})\right)$$
and so in our case, $a_i|p_i \sim N(\rho p_i, 1-\rho^2)$. Also, since the payoffs are independent across projects, we get that
$$a_1-a_2|p_1,p_2 \sim N(\rho (p_1-p_2), 2(1-\rho^2))$$
Using this, we get that $$\mathbb{P}[a_1-a_2\geq 0|p_1, p_2]=\Phi\left(\frac{\rho(p_1-p_2)}{\sqrt{2(1-\rho^2)}}\right)$$ where $\Phi$ is the standard normal cdf. 

Plugging this into the equation, we get that the optimal cutoff satisfies
$$\mathbb{E}\left[(p_1-p_2)\Phi\left(\frac{\rho(p_1-p_2)}{\sqrt{2(1-\rho^2)}}\right) \middle\vert p_1=c(\rho)\right] = 0 $$


We can find that the expectation is equal to
$$p_1 \Phi \left(\dfrac{\rho p_1}{\sqrt{2-\rho^2}} \right)+\dfrac{\rho}{\sqrt{2-\rho^2}}\phi \left(\dfrac{\rho p_1}{\sqrt{2-\rho^2}} \right)$$ and letting
$t=\dfrac{\rho}{\sqrt{2-\rho^2}}$, we have that the optimal cutoff $c$ is implicitly defined by
$$c\Phi(tc)+t\phi(tc)=0 $$
where $\Phi$ and $\phi$ represent the standard normal cdf and pdf respectively. Observe that $t \in [-1,1]$ and is increasing in $\rho$. 


Now letting  $F(c,t)=c\Phi(tc)+t\phi(tc)$, we can use the implicit function theorem to get that 
\begin{align*}
    c'(t)&=-\dfrac{F_t}{F_c}\\
    &=-\dfrac{c^2\phi(tc)+\phi(tc)-t^2c^2\phi(tc)}{\Phi(tc)+tc\phi(tc)-t^3c\phi(tc)}\\
    &=-\dfrac{\phi(tc)\left(1+c^2(1-t^2)\right)}{\Phi(tc)+tc\phi(tc)(1-t^2)}\\
    &=-\dfrac{\phi(tc)\left(1+c^2(1-t^2)\right)}{\Phi(tc)\left(1-c^2(1-t^2)\right)}
\end{align*}


Note that $c(0)=0$ and so $c'(0)=-2\phi(0)<0$. Also observe from the above expression that $c'(t)$ is never $0$ and so it follows from the smoothness of $c$ that $c'(t)<0$ for all $t\in (-1,1)$. Thus, we have that the optimal cutoff is decreasing in $\rho$. 
\end{appendixproof}

The proof proceeds by applying the condition obtained in Theorem ~\ref{thm:n=2_optimal_is_cutoff} for the case of the bivariate normal distribution. Using formulas for integrals of normal cdfs from \citet{owen_table_1980}, we obtain the condition that the optimal cutoff must satisfy $c\Phi(tc)+t\phi(tc)=0$ where $t=\frac{\rho}{\sqrt{2-\rho^2}}$. We can then differentiate this equation and get that $c'(0)<0$. The smoothness of $c$ and the fact that $c'(t)$ is never zero implies that $c'(t)<0$ for all $t$. The formal proof is in the appendix.

\begin{figure}[H]
  \centering
 \includegraphics[width=10cm,height=8cm]{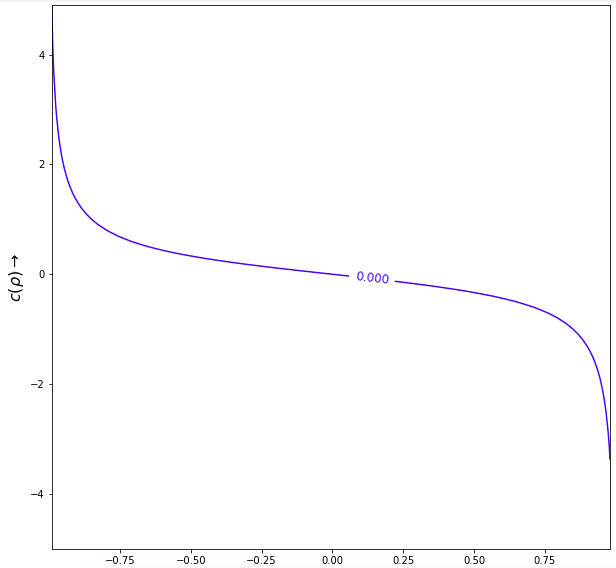}
 \caption{The optimal cutoff for maximizing expected profit decreases as the correlation between principal and agent payoffs increases.}
\end{figure}

\subsection{N projects}

In this subsection, we discuss the case of $N>2$ projects. We believe that a cutoff mechanism should continue to be optimal but we haven't been able to prove it yet. Motivated by the result for two projects, we conjecture that a cutoff mechanism is optimal for an arbitrary number of projects and consider the problem of finding the optimal cutoff mechanism.\\

The following theorem gives the optimal cutoff mechanism when the principal agent payoffs are i.i.d uniform.

\begin{theorem}
\label{thm:optimal_cutoff_value}
For $N$ projects and $F$ uniform on $[0,1]^2$, the optimal cutoff mechanism has a single cutoff $c(N)$ that is defined by the equation 
$$ N(1-c)(1 - c + c^N) = 1 - c^N.$$  The principal's expected utility from the corresponding optimal cutoff mechanism is given by $$\frac{1}{2} + \frac{c(N)}{2}(c(N) - c(N)^N)$$
\end{theorem}

\begin{appendixproof}
\textit{of Theorem \ref{thm:optimal_cutoff_value}}

For any arbitrary cutoff mechanism with cutoffs $c=(c_1,c_2...c_{N-1})$, we compute the expected utility of the principal. For the below expressions, consider $c_N=0$.

\begin{align*}
    EU_p(c)= \sum_{i=1}^N \frac{(1+c_i)}{2} \mathbb{P}(d=i) \\
\end{align*}

Note that $\mathbb{P}(d=i)=(1-c_i)\mathbb{P}(d=N|c_i=0)$. That is, conditional on $p_i \geq c_i$, the probability that the decision is $i$ is the same as the probability that the decision is $N$ when the cutoff $c_i=0$ and the remaining cutoffs are the same.  To find the probability that the last project $N$ is chosen, we condition on its rank which is defined in terms of $a_i$s. That is, the rank of $N$ is $k$ if there are exactly $k-1$ projects with higher $a_i$s.

\begin{align*}
    \mathbb{P}(d=N)&=\sum_{k=1}^{N} \mathbb{P}(\text{rank of } N =k)\mathbb{P}(d=N|\text{ rank of } N =k)\\
    &= \frac{1}{N} \sum_{k=1}^N \mathbb{P}(d=N|\text{ rank of } N =k)\\
    &= \frac{1}{N} \sum_{k=1}^N \sum_{S \subset [N-1]: |S|=k-1} \frac{\Pi_{i \in S} c_i}{ {N-1\choose k-1}} \\
\end{align*}

 We first argue that cutoffs have to be interior in the optimal cutoff mechanism. Suppose $d$ is a cutoff mechanism with cutoffs  $c=(c_1,c_2...c_{N-1})$ and $c_i=1$. In this case, let $u^*$ denote the expected utility of the principal. Note that $u^* < 1$. Define a new cutoff mechanism in which all cutoffs remain the same except for $i$ which now has the cutoff $u^*$. Now, with probability $u^*$, the principal gets $u^*$ and with probability $1-u^*$, the principal gets a convex combination of $u^*$ and $\frac{1+u^*}{2}>u^*$. This means that his expected payoff under the new mechanism is $>u^*$. Therefore, an optimal mechanism cannot have the cutoff 1. Now, suppose there is an $i \in [N-1]$ which has a cutoff of 0. Observe that the expected utility from any arbitrary project conditional on being chosen is $\frac{1+c}{2} \geq 0.5$. Now, consider increasing the cutoff to $c_i=\frac{1}{2}$ in the new mechanism while keeping every other cutoff the same. For every $p_i<\frac{1}{2}$, under the old mechanism, the principal's payoff was some convex combination of $p_i$ and some $k\geq 1/2$. Under the new mechanism, it is just $k>p_i$. For $p_i \geq \frac{1}{2}$, the new mechanism is identical to the old one. Therefore, an optimal mechanism cannot have a cutoff of 0. Thus, we know that in the optimal cutoff mechanism, $c_i \notin \{0,1\}$ for any $i \in [N-1]$.   \\

Now suppose that the  mechanism $d$ is such that there exist $i,j$ with $c_i > c_j$. Define $t$ so that $\bar{c}+t=c_i$ and $\bar{c}-t=c_j$. From the above calculations, we know that if we write $EU_p(c)$ in expanded form and plug in $c_i=\bar{c}+t$ and $c_j=\bar{c}-t$, we get a polynomial that is at most cubic in $t$. This is because we get a term that is at most quadratic in $t$ for $\mathbb{P}(d=k)$ for any $k \in [n]$ and in the expected utility calculation, we multiply that with $\frac{1+c_k}{2}$. Note that by the symmetry of the projects, the principal should get the same expected utility if we changed $t$ to $-t$. Therefore, the polynomial should be of the form $at^2+b$. Now, if $a$ is $>0$ or $<0$, the principal gains from increasing or decreasing $t$ which is possible since we know that the solution is interior and $c_i>c_j$. Therefore, $d$ cannot be optimal in either case. When $a=0$, the principal is indifferent to increasing $t$ till one of the cutoffs reaches an extreme of $1$ or $0$ which we know cannot be optimal. Therefore, a cutoff mechanism with different cutoffs cannot be optimal.\\

The above discussion implies that the solution to the optimization problem has to be a single cutoff mechanism. Let $c$ be the single cutoff. Using the above calculations, we have that

\begin{align*}
    \mathbb{P}(d=N)&= \frac{1}{N} \sum_{k=1}^N \sum_{S \subset [N-1]: |S|=k-1} \dfrac{\Pi_{i \in S} c_i}{ {N-1\choose k-1}} \\
    &= \frac{1}{N}  \sum_{k=1}^N c^{k-1}\\
    &= \frac{1}{N}\frac{1-c^N}{1-c}\\
\end{align*}

Therefore, 

\begin{align*}
    EU_p(c)&=\frac{1}{2} \mathbb{P}(d=N)+\frac{1+c}{2} \mathbb{P}(d \neq N)\\
    &=\frac{1}{2}+\frac{c}{2} \left(1-\frac{1-c^N}{N(1-c)}\right)\\
\end{align*}

Differentiating with respect to $c$ gives the desired optimal cutoff mechanism defined by single cutoff $c(N)$.

\begin{align*}
    \dfrac{\partial EU_p(c)}{\partial c} &= \dfrac{1}{2}\left(1-\frac{1-c^N}{N(1-c)}\right)-\frac{c}{2N} \left(\frac{(1-c)(-Nc^{N-1})+(1-c^N)}{(1-c)^2}\right) \\
    &= \frac{1}{2}+\dfrac{c^N}{2(1-c)}-\dfrac{1-c^N}{2N(1-c)^2}\\
\end{align*}

Setting it equal to zero gives us:

$$N(1-c)(1-c+c^N)=1-c^N$$ 

Plugging in the expected utility expression gives us the maximum utility of the principal in the class of cutoff mechanisms: $\frac{1}{2}+\frac{c(N)}{2}(c(N)-c(N)^N)$.
\end{appendixproof}

Note that since we are maximizing a continuous function over a compact space, a solution exists. The rest of the proof proceeds in three steps. In step 1, we show that in any cutoff mechanism which has a cutoff that is $0$ or $1$ cannot be optimal. That is, the solution must be interior. In step 2,  we show that a cutoff mechanism with different cutoffs cannot be optimal. Finally, we maximize the principal's utility with respect to the single cutoff $c$ to get the optimal cutoff mechanism. The proof is relegated to the appendix.
While it is hard to obtain an exact closed form solution for the optimal single cutoff, we show that it has the following asymptotic property.

\begin{lem}
\label{lem:explicit-cutoff}
The optimal single cutoff $c(N)= 1 - \frac{1}{\sqrt{N}} + o\left(\frac{1}{\sqrt{N}}\right).$
\end{lem}

\begin{appendixproof}
\textit{of lemma \ref{lem:explicit-cutoff}}

Let
\begin{align*}
    \phi_N(c) = N(1-c)(1 - c + c^N) - (1 - c^N).
\end{align*}
Then for any $\alpha > 0$,
\begin{align*}
    \lim_{N \to \infty}{\phi_N \left(1 - \frac{\alpha}{\sqrt{N}}\right)} = \alpha^2 - 1,
\end{align*}
and so for all sufficiently large $N$, the quantity
\begin{align*}
    \phi_N \left(1 - \frac{\alpha}{\sqrt{N}}\right)
\end{align*}
is positive if and only if $\alpha > 1$ and negative if and only if $\alpha < 1$. Hence, for any $\epsilon > 0$, it follows that the unique root $c(N)$ of the equation from Theorem \ref{thm:optimal_cutoff_value} satisfies
\begin{align*}
    (1-\epsilon) \frac{1}{\sqrt{N}} \leq 1 - c(N) \leq (1+\epsilon) \frac{1}{\sqrt{N}}.
\end{align*}
\end{appendixproof}

It follows from Lemma \ref{lem:explicit-cutoff} that as $N \to \infty$, the optimal single cutoff $c(N) \to 1$ and the expected utility of the principal $\to 1$.
\section{Maximizing probability of best project}

In this section, we consider the objective of maximizing the probability of choosing the best project for the principal $\mathbb{P}[p_{d(p,a)} \geq p_j \text{ for all } j]$.  We'll only focus on the case of $N=2$ projects for this part. 

\subsection{2 projects}

Let's consider the case of two projects and table mechanisms with project 2 always on the table. 

\begin{theorem}
\label{thm:n=2_optimal_is_cutoff_prob}

For two projects with $(p_i,a_i) \sim F$, the optimal truthful mechanism is a cutoff mechanism. The optimal cutoff $c_1$ is defined by
$$\mathbb{E}\left[(\mathbb{I}_{p_1\geq p_2}-\mathbb{I}_{p_2>p_1})\mathbb{I}_{a_1\geq a_2}\vert p_1=c_1\right]=0$$
or more simply
$$\mathbb{P}[p_1 \geq p_2\vert p_1=c_1,a_1 \geq a_2]=\frac{1}{2}$$

\end{theorem}

\begin{proof}
The argument for why cutoff is optimal is exactly the same as in the proof of Theorem ~\ref{thm:n=2_optimal_is_cutoff}. We now derive the optimal cutoff. 
The principal can either 
\begin{itemize}
    \item make project 1 available and get $\mathbb{I}_{p_1\geq p_2}\mathbb{I}_{a_1\geq a_2}+\mathbb{I}_{p_2> p_1}\mathbb{I}_{a_2> a_1}$
    \item or not make it available and get $\mathbb{I}_{p_2>p_1}$
\end{itemize}

Since the principal can only make the decision based on value of $p_1$, it will chose to make project 1 available if and only if 
\begin{align*}
    &\mathbb{E}[\mathbb{I}_{p_1\geq p_2}\mathbb{I}_{a_1\geq a_2}+\mathbb{I}_{p_2>p_1}\mathbb{I}_{a_2> a_1}\vert p_1]\geq  \mathbb{E}[\mathbb{I}_{p_2>p_1} \vert p_1]\\
    \iff &\mathbb{E}\left[(\mathbb{I}_{p_1\geq p_2}-\mathbb{I}_{p_2>p_1})\mathbb{I}_{a_1\geq a_2}\vert p_1\right]\geq 0
\end{align*}

At the cutoff, the principal must be indifferent between making or not making project 1 available. This gives the desired condition. 
\end{proof}
In the special case where the principal and agent payoffs are independent, we get the following corollary.
\begin{corollary}
Suppose $F$ is such that the principal and agent payoffs $(p_i,a_i)$ are independent. Then the optimal cutoff is given by $c_1=Med[p_1]$.
\end{corollary}

\subsection{Ally principle}
Assume that the principal agent payoffs for each project is bivariate normal $(p_i,a_i) \sim N(0,0,1,1,\rho)$  and are drawn i.i.d. across projects. Now the question is whether we can say something systematic about $c(\rho)$, the optimal cutoff as a function of the correlation $\rho$.

\begin{theorem}
\label{cutoff2}
For $N=2$ projects with $(p_i, a_i) \sim N(0,0,1,1,\rho)$, the optimal cutoff $c(\rho)$ is given by the equation $$\dfrac{\mathbb{P}[X \leq c, Y \leq tc]}{\mathbb{P}[Y \leq tc]}=\frac{1}{2}$$ where $X,Y \sim N(0,0,1,1,t)$ and $t=\dfrac{\rho}{\sqrt{2-\rho^2}}$. 

\end{theorem}

\begin{appendixproof}
\textit{of Theorem ~\ref{cutoff2}}

Given the distributional form, we have

$$a_1-a_2|p_1 \sim N(\rho p_1, 2-\rho^2)$$ and therefore, 
$$p_2 | a_1-a_2, p_1 \sim N\left(\dfrac{\rho^2p_1-\rho (a_1- a_2)}{2-\rho^2}, \dfrac{2(1-\rho^2)}{2-\rho^2}\right)$$

This gives
$$\mathbb{P}\left[p_2 \leq p_1|p_1,a_1-a_2\right]=\Phi \left(\dfrac{2p_1(1-\rho^2)+\rho(a_1-a_2)}{\sqrt{2(1-\rho^2)(2-\rho^2)}}\right)$$

Then, 
\begin{align*}
    \mathbb{P}[p_2 \leq p_1\vert p_1,a_1 \geq a_2]&=\frac{1}{\mathbb{P}[a_1\geq a_2|p_1]}\int_0^\infty \mathbb{P}[p_2 \leq p_1\vert p_1,a_1- a_2=x]f(a_1-a_2=x|p_1)dx \\
    &=\frac{1}{\Phi\left(\dfrac{\rho p_1}{\sqrt{2-\rho^2}}\right)}\int_0^\infty \Phi \left(\dfrac{2p_1(1-\rho^2)+\rho x}{\sqrt{2(1-\rho^2)(2-\rho^2)}}\right)\frac{e^{-\frac{1}{2}\left(\frac{x-\rho p_1}{\sqrt{2-\rho^2}}\right)^2}}{\sqrt{2-\rho^2}\sqrt{2\pi}}dx\\
    &=\frac{1}{\Phi\left(\dfrac{\rho p_1}{\sqrt{2-\rho^2}}\right)}\int_\frac{-\rho p_1}{\sqrt{2-\rho^2}}^\infty \Phi \left(\dfrac{2p_1(1-\rho^2)+\rho(t\sqrt{2-\rho^2}+\rho p_1)}{\sqrt{2(1-\rho^2)(2-\rho^2)}}\right)\phi(t)dt\\
    &=\frac{1}{\Phi\left(\dfrac{\rho p_1}{\sqrt{2-\rho^2}}\right)}\int_\frac{-\rho p_1}{\sqrt{2-\rho^2}}^\infty \Phi \left(\dfrac{p_1\sqrt{2-\rho^2}+\rho t}{\sqrt{2(1-\rho^2)}}\right)\phi(t)dt\\
    &=\mathbb{E}\left[\Phi\left(\dfrac{p_1\sqrt{2-\rho^2}+\rho t}{\sqrt{2(1-\rho^2)}}\right) \bigg\vert t \geq \dfrac{-p_1\rho}{\sqrt{2-\rho^2}}\right]
\end{align*}

Observe that if $f(p_1,\rho)$ is the above expectation, then we have $f(p_1,\rho)+f(-p_1,-\rho)=1$. Thus, we can conclude that $c(-\rho)=-c(\rho)$ for all $\rho \in [0,1]$. So let's focus on $\rho>0$ and try to argue that the optimal cutoff must be decreasing in $\rho$. 


The expectation above equals
$$\dfrac{\mathbb{P}[X \leq p_1, Y \leq t p_1]}{\mathbb{P}[Y \leq t p_1]}$$ where $X,Y \sim N(0,0,1,1,t)$ and $t=\dfrac{\rho}{\sqrt{2-\rho^2}}$ and this completes the proof.
\end{appendixproof}

Again, we use the definition of optimal cutoff from Theorem ~\ref{thm:n=2_optimal_is_cutoff_prob} and apply it to the bivariate normal case. Then, using the integral formulas from \citet{owen_table_1980}, we show that the optimal cutoff in this case is defined by $\dfrac{\mathbb{P}[X \leq c, Y \leq t c]}{\mathbb{P}[Y \leq t c]}$ where $X,Y \sim N(0,0,1,1,t)$ and $t=\dfrac{\rho}{\sqrt{2-\rho^2}}$.  We haven't been able to show that this result implies that the optimal cutoff is decreasing in $\rho$, but the following contour plot from Python suggests that it is:

\begin{figure}[H]
  \centering
 \includegraphics[width=10cm,height=8cm]{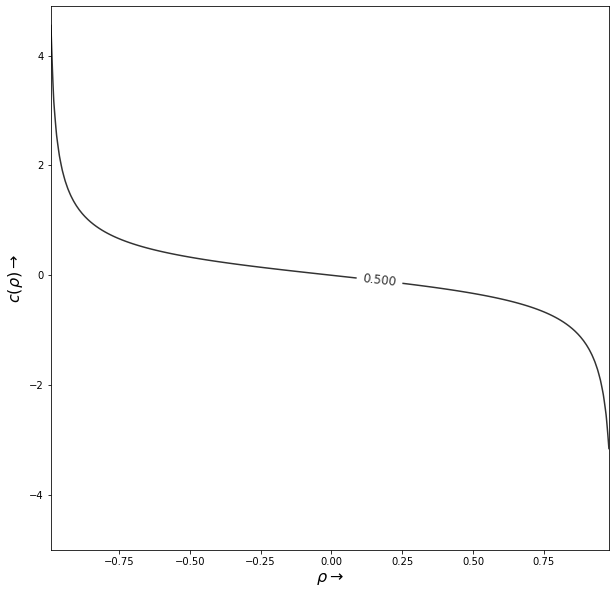}
 \caption{The optimal cutoff for maximizing probability of choosing better project as a function of correlation (contour plot made in Python).}
\end{figure}

\section{Remarks}

\begin{enumerate}[leftmargin=*]
    \item \textbf{Comparison of optimal mechanisms for the two objectives} We wanted to see how the optimal cutoffs for the two objectives compare when payoffs are bivariate normal. To do so, we plotted the two optimal cutoff curves together and interestingly found that the curves coincide. While we haven't been able to formally prove that the solutions coincide, we conjecture that they do.
    \begin{conjecture}
    With $N=2$ projects and payoffs $(p_i,a_i)$ drawn iid from bivariate normal $N(0,0,1,1,\rho)$, the optimal mechanism for maximizing principal's expected profit and that for maximizing probability of choosing better project is the same cutoff mechanism with the cutoff $c(\rho)$ defined by $c\Phi(tc)+t\phi(tc)=0$ where $t=\frac{\rho}{\sqrt{2-\rho^2}}$
    \end{conjecture}
    
    \item \textbf{Delegation interpretation of the optimal mechanism} The simplest implementation of the optimal cutoff mechanism has a nice delegation interpretation. The principal selects a cutoff profit and a default project and delegates the project choice to the agent, in the sense that the agent can select either the default project or a project which meets the cutoff profit, and the principal signs off on the final decision. Under this delegation, the agent chooses his favorite project among those that meet the cutoff and the default project. Note in particular that this implementation only requires the agent to report information about the chosen project. This is outcome-equivalent to the cutoff mechanism. We note that many instances of ``cutoff mechanisms" with flavors similar to ours have appeared in the literature, but the optimality of such mechanisms has been driven by the assumption of ex-post verifiability (\citet{ben2014optimal}, \citet{mylovanov2017optimal} \citet{armstrong2010model}). In particular, in most previous models the principal's ability to punish in the case of a misreport is tantamount to the assumption that the agent cannot lie. Here, we offer an alternative way of rationalizing such cutoff mechanisms via the no-overselling (or more generally, interim partial verifiability) constraint, which alters the agents incentives but \textit{not} by threatening the agent in the case of a misreport. To help elucidate this point, we make the following observation. If our model were altered so that the agent had an unconstrained message space, but the principal were required to take the default project in case the agent should oversell any of the projects, then all of our results would carry over.

\end{enumerate}
\section{Conclusion}
We consider a principal agent project selection problem with asymmetric information. When the agent can lie arbitrarily, we find that the principal cannot gain anything from commitment power and may as well choose a project at random. In contrast, if the principal can identify or induce partial verifiability in the environment so that the agent cannot oversell any of the projects, then a simple cutoff mechanism is  optimal for the case of two projects, both for maximizing expected profit and for maximizing probability of choosing better project. 
In the particular case where payoffs are bivariate normal, we find that the optimal cutoff is decreasing in the correlation between payoffs and thus, our model provides evidence in favor of the ally principle which says that the principal grants more leeway to an agent who shares its preferences. We conjecture that our results for the case of two project extend to settings with more than two projects as well. \\

\nocite{*}

\bibliographystyle{ecta}

\bibliography{refs}

\end{document}